\documentclass{llncs}

\usepackage{latexsym,amsfonts}
\usepackage{amsmath}
\usepackage{amssymb}
\usepackage{array}

\usepackage{graphicx,color}
\usepackage{latexsym,amsfonts}
\usepackage{amsmath}
\usepackage{amssymb}
\usepackage{bbm}

\usepackage{algorithmic}
\usepackage{algorithm}
\usepackage{tikz}
\usepackage{url}

\usetikzlibrary{calc}

\newcounter{lpnumber} 

\newcommand{\Comp}{\mathcal{C}}
\newcommand{\vote}{\mathsf{vote}}
\newcommand{\wt}{\mathsf{wt}}

\newtheorem{new-claim}{Claim}

\begin{document}
\title{The Popular Roommates problem}
\author{Telikepalli Kavitha}
\institute{Tata Institute of Fundamental Research, Mumbai, India\\
\email{kavitha@tcs.tifr.res.in}}
\maketitle
\pagestyle{plain}

\begin{abstract}
  We consider the {\em popular matching} problem in a roommates instance $G = (V,E)$ with strict preference lists. While popular matchings always
  exist in a bipartite instance, they need not exist in a roommates instance. The complexity of the popular matching problem in a roommates
  instance has been an open problem for several years and we prove its NP-hardness here.
  A sub-class of max-size popular matchings called {\em dominant matchings} has been well-studied in bipartite graphs. We show that 
  the dominant matching problem in $G = (V,E)$ is also NP-hard and this is the case even when $G$ admits a stable matching.
\end{abstract}

\section{Introduction}
\label{intro}
We consider a matching problem in a graph $G = (V,E)$ (need not be complete) where each vertex $u \in V$ ranks its neighbors in a strict order of preference. Such a
graph $G$ is usually referred to as a {\em roommates instance}. A matching $M$ is stable if there is no {\em blocking edge}
with respect to $M$, i.e., there is no pair $(a,b)$ such that both $a$ and $b$ prefer each other to their respective assignments in $M$.

Stable matchings always exist when $G$ is bipartite~\cite{GS62}, however there are simple roommates instances  that do not admit any stable matching.
The problem of deciding whether a stable matching exists or not in $G$ is  the {\em stable roommates} problem.
There are several polynomial time algorithms~\cite{Irv85,Sub94,TS98} to solve the stable roommates problem.
Here we consider a notion called {\em popularity} that is more relaxed than stability.

\subsection{Popular Matchings}
The notion of popularity was introduced by G\"ardenfors~\cite{Gar75} in 1975. 
We say a vertex $u$ {\em prefers} matching $M$ to matching $M'$ if either (i)~$u$ is matched in $M$
and unmatched in $M'$ or (ii)~$u$ is matched in both $M, M'$ and $u$ prefers $M(u)$ to $M'(u)$. 
For any two matchings $M$ and $M'$, let $\phi(M,M')$ be the number of vertices that prefer $M$ to $M'$.

\begin{definition}
\label{pop-def}
A matching $M$ is {\em popular} if  $\phi(M,M') \ge  \phi(M',M)$ for every matching $M'$ in $G$, 
i.e., $\Delta(M,M') \ge 0$ where $\Delta(M,M') = \phi(M,M') -  \phi(M',M)$.
\end{definition}

Thus there is no matching $M'$ that would defeat a popular matching $M$ in an election between $M$ and $M'$, where each vertex casts a vote
for the matching that it prefers. Since there is no matching where more vertices are {\em better-off} than in a popular matching, a popular matching
can be regarded a ``globally stable matching''.

It is easy to show that every stable matching is popular~\cite{Gar75}. Since popularity is a relaxation of stability, popular matchings may exist
in roommates instances that admit no stable matchings (see the instance on 4 vertices $d_0,d_1,d_2,d_3$ on the left of Fig.~\ref{level1:example}).
Here we are interested in the complexity of the {\em popular roommates}
problem, i.e., the problem of deciding if $G = (V,E)$ admits a popular matching or not. This has been an open problem for almost a
decade~\cite{BIM10} and we show the following result here.\footnote{Very recently, this hardness result also appeared in \cite{GMSZ18} on the arxiv; our results were obtained independently and our proofs are different.}

\begin{theorem}
  \label{main-thm}
  Given a roommates instance $G = (V,E)$ with strict preference lists, the problem of deciding if $G$ admits a popular matching or not is NP-hard.
\end{theorem}

Popular matchings always exist in a bipartite instance, since stable matchings always exist here.
Popular matchings have been well-studied in bipartite graphs, in particular,
a subclass of max-size popular matchings called {\em dominant matchings} is well-understood~\cite{CK16,HK11,Kav12}.

\begin{definition}
  A popular matching $M$ is dominant in $G$ if $M$ is more popular than any larger matching in $G$, i.e.,  $\Delta(M,M') > 0$ for any matching
  $M'$ such that $|M'| > |M|$.
\end{definition}  

Dominant matchings always exist in a bipartite instance and such a matching can be computed in linear time~\cite{Kav12}.
We consider the dominant matching problem in a roommates instance and show the following result.

\begin{theorem}
  \label{second-thm}
  Given a roommates instance $G = (V,E)$ with strict preference lists, the problem of deciding if $G$ admits a dominant matching or not
  is NP-hard. Moreover, this hardness holds even when $G$ admits a stable matching.
\end{theorem}  
  
\subsection{Background and Related work}
The first polynomial time algorithm for the stable roommates problem was by Irving~\cite{Irv85} in 1985. 
A characterization of roommates instances that admit stable matchings was given in \cite{Tan91} and 
new polynomial time algorithms for the stable roommates problem were given in \cite{Sub94,TS98}. As mentioned earlier,
G\"ardenfors~\cite{Gar75} introduced the notion of popularity in the stable matching problem in bipartite instances.

Algorithmic questions for popular matchings were initially studied in the one-sided preference lists model: here only one side of the bipartite instance
has preferences over its neighbors. Popular matchings need not always exist in this model and there is an efficient algorithm~\cite{AIKM07} to determine if
a given instance admits one. Popular {\em mixed} matchings always exist here~\cite{KMN09} and
such a mixed matching can be computed in polynomial time via linear programming.

In the stable matching problem in bipartite instances (the two-sided preference lists model), popular matchings always exist and 
a max-size popular matching can be computed efficiently~\cite{HK11,Kav12}. These algorithms always compute dominant matchings --- it was shown in \cite{CK16}
that dominant matchings are essentially stable matchings in a larger bipartite graph.
When ties are allowed in preference lists, 
the problem of deciding if a popular matching problem exists or not is NP-hard~\cite{BIM10,CHK15}.

It was shown in \cite{HK17} that the problem of computing a max-weight popular matching in a roommates instance with edge weights is NP-hard. This was
strengthened in \cite{Kav18} to show that  the problem of computing a max-size popular matching in a roommates instance is NP-hard. An efficient
algorithm was also given in \cite{Kav18} to compute a {\em strongly dominant} matching in a roommates instance. Strongly dominant matchings are a
subclass of dominant matchings; interestingly, in bipartite instances, dominant and strongly dominant are equivalent notions~\cite{Kav12}.

It was shown in \cite{HK13} that every roommates instance $G = (V,E)$ admits a matching whose {\em unpopularity factor} is $O(\log|V|)$ and it is NP-hard to
compute a least unpopularity matching in $G$. The complexity of the popular roommates problem was stated as an open problem in several papers/books~\cite{BIM10,Cseh17,HK13,HK17,Man2013}.

\medskip

\noindent{\bf Techniques.}
We use properties of popular matchings in bipartite instances here
--- in particular, we use the LP framework of popular matchings that was initiated in \cite{KMN09}. Every popular matching $M$ in a bipartite instance $H$
is a max-weight perfect matching in a related graph $\tilde{H}$ and an optimal solution to the dual LP (dual to the max-weight perfect matching LP) is
a {\em witness} to the popularity of $M$. It is known that
a matching in $H$ is popular if and only if it has a witness $\vec{\alpha} \in \{0,\pm 1\}^n$, where $n$ is the number of vertices.

A stable matching in our roommates instance $G$ will correspond to a matching in $H$ with  $\vec{0}$ as a witness
and a strongly dominant matching in $G$ will correspond to a matching in $H$ with a witness
$\vec{\alpha}$ such that $\alpha_u \in \{\pm 1\}$ for all matched vertices $u$. We show a reduction from 1-in-3 SAT to the popular roommates
problem via the problem of deciding if a {\em desired popular matching} exists in the bipartite instance $H$;
such a matching is constrained to have a certain witness in $\{0, \pm 1\}^n$ which will prove its hardness.

\medskip

\noindent{\em Organization of the paper.} Section~\ref{prelims} contains an overview of the LP framework of popular matchings in bipartite instances.
Section~\ref{sec:hardness} outlines the reduction from 1-in-3 SAT to the popular roommates problem and Section~\ref{sec:thm3-proof} has more details of
our hardness reduction. Section~\ref{sec:domn} shows NP-hardness for dominant matchings.

\section{Preliminaries}
\label{prelims}
This section is an overview of the LP framework of popular matchings in bipartite graphs from \cite{KMN09} along with some results
from \cite{Kav16,Kav18}.
Let $H = (A \cup B, E_H)$ be a bipartite instance with strict preference lists and 
let $\tilde{H}$ be the graph $H$ augmented with {\em self-loops}, i.e., it is assumed that every vertex is its own last choice.
Let $M$ be any matching in $H$. Corresponding to $M$,
there is a perfect matching $\tilde{M}$ in $\tilde{H}$ defined as follows: $\tilde{M} = M \cup \{(u,u): u$ is left unmatched in $M\}$. 

We now define an edge weight function $\wt_M$ in $\tilde{H}$. For any vertex $u$ and neighbors $v, v'$ in $\tilde{H}$, let
$\vote_u(v,v')$ be 1 if $u$ prefers $v$ to $v'$,
it is -1 if $u$ prefers $v'$ to $v$, else it is 0 (i.e., $v = v'$). The function $\wt_M$ is defined as follows:
\[ \wt_M(u,v) \ = \ \vote_u(v,\tilde{M}(u)) + \vote_v(u,\tilde{M}(v)) \ \ \ \mathrm{for}\ (u,v) \in E_H.\]

Thus $\wt_M(u,v) \in \{0, \pm 2\}$.
We need to define $\wt_M$ on self-loops as well: for any $u \in V$, let $\wt_M(u,u) = 0$ if $u$ is unmatched in $M$, else
let $\wt_M(u,u) = -1$. Thus $\wt_M(u,u)$ is $u$'s vote for itself versus $\tilde{M}(u)$.

It is easy to see that for any matching $N$ in $H$, $\Delta(N,M) = \wt_M(\tilde{N})$.
Thus $M$ is popular if and only if every perfect matching in $\tilde{H}$ has weight at most 0. Let $n = |A \cup B|$.

\begin{theorem}[\cite{KMN09}]
\label{thm:witness}
  Let $M$ be any matching in $H = (A \cup B, E_H)$. The matching $M$ is popular if and only if there exists a vector $\vec{\alpha} \in \mathbb{R}^n$
  such that $\sum_{u \in A \cup B}\alpha_u = 0$ and
  \begin{eqnarray*}
    \alpha_{a} + \alpha_{b} & \ \ \ge \ \ & \wt_{M}(a,b) \ \ \ \forall\, (a,b)\in E_H\\
    \alpha_u & \ \ \ge \ \ & \wt_M(u,u) \ \ \ \forall\, u\in A \cup B.
  \end{eqnarray*}  
\end{theorem}  

The vector $\vec{\alpha}$ will be an optimal solution to the LP that is dual to the max-weight perfect matching LP in $\tilde{H}$
(with edge weight function $\wt_M$). For any popular matching $M$, a vector $\vec{\alpha}$ as given in Theorem~\ref{thm:witness}
will be called a {\em witness} to $M$. The following lemma will be useful to us.

\begin{lemma}[\cite{Kav16}]
  Any popular matching in $H = (A \cup B, E_H)$ has a witness in $\{0,\pm 1\}^n$.
\end{lemma}

Call any $e \in E_H$ a {\em popular edge} if there is some popular matching in $H$ that contains $e$.
Let $M$ be a popular matching in $H$ and let $\vec{\alpha} \in \{0,\pm 1\}^n$ be a witness of $M$.

\begin{lemma}[\cite{Kav18}]
\label{prop0}
If $(a,b)$ is a popular edge in $H$, then $\alpha_a + \alpha_b = \wt_M(a,b)$.
If $u$ is an unstable vertex in $H$ then $\alpha_u = 0$ if $u$ is left unmatched in $M$, else $\alpha_u = -1$.
\end{lemma}

The popular  subgraph $F_H$ is a useful subgraph of $H$ defined in \cite{Kav18}.

\begin{definition}
\label{def:popular-subgraph}
  The {\em popular subgraph} $F_H = (A \cup B, E_F)$ is the subgraph of $H = (A \cup B, E_H)$
whose edge set $E_F$ is the set of popular edges in $E_H$.
\end{definition}

The graph $F_H$ need not be connected. Let $\Comp_1,\ldots,\Comp_h$ be the various components in $F_H$.

\begin{lemma}[\cite{Kav18}]
\label{prop1}
For any connected component $\Comp_i$ in $F_H$, either $\alpha_u = 0$ for all vertices $u \in \Comp_i$ or 
$\alpha_u \in \{\pm 1\}$ for all vertices $u \in \Comp_i$. Moreover, if $\Comp_i$ contains one or 
more unstable vertices, either all these unstable vertices are matched in $M$ or none of them is 
matched in $M$.
\end{lemma}

The following definition marks the state of each connected component  $\Comp_i$ in $F_H$ as ``stable'' or ``dominant'' in $\vec{\alpha}$
--- this classification will be useful to us in our hardness reduction.

\begin{definition}
  \label{def:stab-domn}
A connected component $\Comp_i$ in $F_H$ is in {\em stable state} in $\vec{\alpha}$ if  $\alpha_u = 0$ for all vertices $u \in \Comp_i$.
Similarly, $\Comp_i$ in $F_H$ is in {\em dominant state} in $\vec{\alpha}$ if  $\alpha_u \in \{\pm 1\}$ for all vertices $u \in \Comp_i$.
\end{definition}

\section{Hardness of the popular roommates problem}
\label{sec:hardness}

Our reduction will be from 1-in-3 SAT. Recall that 1-in-3 SAT is the set of 3CNF formulas with no negated variables such that there is a
satisfying assignment that makes {\em exactly one} variable true in each clause. Given an input formula $\phi$, to determine if $\phi$ is
1-in-3 satisfiable or not is NP-hard~\cite{Sch78}.

We will now build a roommates instance $G = (V, E)$. The vertex set $V$ will consist of vertices in 4 levels: levels~0, 1, 2, and 3 along with
5 other vertices $d_0,d_1,d_2,d_3$, and $z$. The vertices $d_0,d_1,d_2,d_3$ form a gadget $D$ (on the left of Fig.~\ref{level1:example}) described below.
The vertex $d_0$ will be the last choice neighbor of vertex $v$ for every $v \in V \setminus \{z\}$.

Vertices in $V \setminus \{d_0,d_1,d_2,d_3,z\}$ are partitioned into gadgets that appear in some level~$i$, for $i \in \{0,1,2,3\}$.
Every edge $(u,v)$ in $G$ where $u, v$ are in $V \setminus \{d_0,z\}$ is either inside a gadget
or between 2 gadgets in consecutive levels. We describe these gadgets below.

\medskip

\noindent{\em Level~1 vertices.} Every gadget in level~1 is a variable gadget.
Corresponding to each variable $X_i$, we will have the gadget on the right of Fig.~\ref{level1:example}. 
The preference lists of the 4 vertices in the gadget corresponding to $X_i$ are as follows:

\begin{minipage}[c]{0.45\textwidth}
			
			\centering
			\begin{align*}
			        &x_i\colon \, y_i \succ y'_i \succ  z \succ \cdots   \qquad\qquad &&  y_i\colon \, x_i \succ x'_i \succ z \succ \cdots \\
                                &x'_i\colon \, y_i \succ y'_i \succ \cdots   \qquad\qquad &&  y'_i\colon \, x_i \succ x'_i \succ \cdots\\
			\end{align*}
\end{minipage}

\smallskip

The vertices in the gadget corresponding to $X_i$ are also adjacent to vertices in the ``clause gadgets'' corresponding to $X_i$:
these neighbors belong to the ``$\cdots$'' part of the preference lists. Note that the order among the vertices in the ``$\cdots$'' part
in the above preference lists does not matter.  Also, $d_0$ is the last choice of each of the above vertices.

\begin{figure}[h]
\centerline{\resizebox{0.46\textwidth}{!}{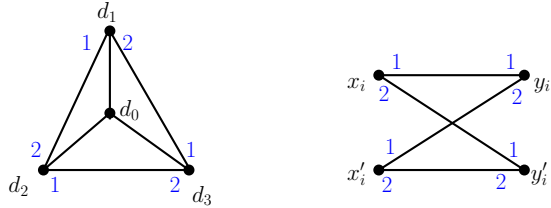}}
\caption{Each of $d_1,d_2,d_3$ is a top choice neighbor for another vertex here and $d_0$ is the last choice of $d_1,d_2,d_3$.
On the right is the gadget corresponding to variable $X_i$: vertex preferences are indicated on edges.}
\label{level1:example}
\end{figure}

\noindent{\em The gadget $D$.}
There will be 4 vertices $d_0,d_1,d_2,d_3$ that form the gadget $D$ (see the left of Fig.~\ref{level1:example}).
The preferences of vertices in $D$ are given below.

\begin{minipage}[c]{0.45\textwidth}
			
			\centering
			\begin{align*}
				&d_1\colon \, d_2  \succ d_3 \succ d_0  \qquad\qquad && d_2\colon \, d_3  \succ d_1 \succ d_0\\
			        &d_3\colon \, d_1 \succ d_2 \succ d_0  \qquad\qquad && d_0\colon \, d_1 \succ d_2 \succ d_3 \succ \cdots \\
			\end{align*}
\end{minipage}

The vertex $d_0$ will be adjacent to all vertices in $G$ other than $z$. The order of other neighbors in $d_0$'s preference list does not matter.
Let $c = X_i \vee X_j \vee X_k$ be a clause in $\phi$. We will describe the gadgets that correspond to $c$.

\begin{figure}[h]
\centerline{\resizebox{0.74\textwidth}{!}{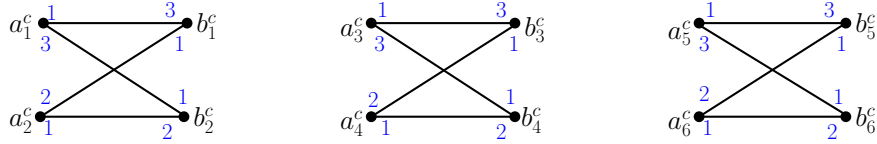}}
\caption{Corresponding to clause $c =  X_i \vee X_j \vee X_k$ we have the above 3 gadgets in level 0. The vertex $a^c_1$'s second choice
  is $y'_j$ and $b^c_1$'s is $x'_k$, similarly, $a^c_3$'s is $y'_k$ and $b^c_3$'s is $x'_i$, also $a^c_5$'s is $y'_i$ and $b^c_5$'s is $x'_j$.}
\label{level0:example}
\end{figure}

\medskip

\noindent{\em Level~0 vertices.} 
There will be three level~0 gadgets, each on 4 vertices, corresponding to clause $c$. See Fig.~\ref{level0:example}.
We describe below the preference lists of the 4 vertices $a^c_1,b^c_1,a^c_2,b^c_2$ that belong to the leftmost gadget.
For the sake of readability, we have dropped the superscript $c$ from these vertices.

\begin{minipage}[c]{0.45\textwidth}
			
			\centering
			\begin{align*}
			        &a_1\colon \, b_1 \succ \underline{y'_j} \succ b_2 \succ \underline{z}  \qquad\qquad && b_1\colon \, a_2 \succ \underline{x'_k} \succ a_1 \succ \underline{z} \\
                                &a_2\colon \, b_2 \succ b_1  \qquad\qquad && b_2\colon \, a_1 \succ a_2 \\
			\end{align*}
\end{minipage}

Though $d_0$ is not explicitly listed in the above preference lists,
recall that $d_0$ is the last choice of each of these vertices. Neighbors that are outside this gadget are underlined.
The preferences of vertices in the other 2 gadgets in level~0 corresponding to $c$ ($a^c_t,b^c_t$ for $t = 3,4$ and $a^c_t,b^c_t$ for $t = 5,6$) are analogous.
We will now describe the three level~2 gadgets corresponding to clause $c$. See Fig.~\ref{level2:example}.

\begin{figure}[h]
\centerline{\resizebox{0.8\textwidth}{!}{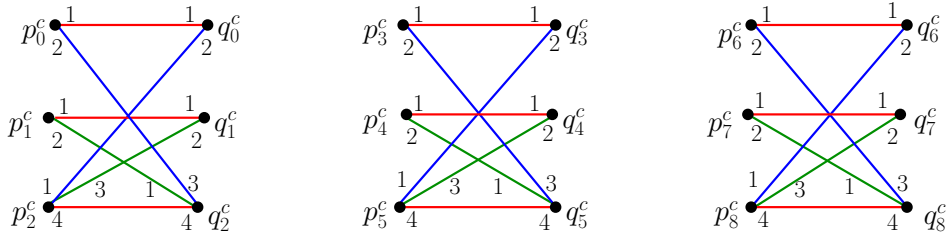}}
\caption{We have the above 3 gadgets in level 2 corresponding to $c =  X_i \vee X_j \vee X_k$ . The vertex $p^c_2$'s second choice is
   $y_j$ and $q^c_2$'s is $x_k$, similarly, $p^c_5$'s is $y_k$ and $q^c_5$'s is $x_i$, similarly $p^c_8$'s is $y_i$ and $q^c_8$'s is $x_j$.}
\label{level2:example}
\end{figure}

\noindent{\em Level~2 vertices.}
There will be three level~2 gadgets, each on 6 vertices, corresponding to clause $c$.
The preference lists of the vertices $p^c_t,q^c_t$ for $0 \le t \le 2$ are described below.
Also, $d_0$ is the last choice of each of these vertices.
For the sake of readability, we have again dropped the superscript $c$ from these
vertices. 

\begin{minipage}[c]{0.45\textwidth}
			
			\centering
			\begin{align*}
			        &p_0\colon \, q_0 \succ q_2  \qquad\qquad && q_0\colon \, p_0 \succ p_2 \succ \underline{z} \succ \underline{s_0}\\
                                &p_1\colon \, q_1 \succ q_2 \succ \underline{z}  \qquad\qquad && q_1\colon \, p_1 \succ p_2 \\
                                &p_2\colon \, q_0 \succ \underline{y_j} \succ q_1 \succ q_2  \qquad\qquad && q_2\colon \, p_1 \succ \underline{x_k} \succ p_0 \succ p_2 \\
			\end{align*}
\end{minipage}

Let us note the preference lists of $p_2$ and $q_2$:
they are each other's fourth choices.
The vertex $p_2$ regards $q_0$ as its top choice, $y_j$ as its second choice, and $q_1$ as its third choice.
The vertex $q_2$ regards $p_1$ as its top choice, $x_k$ as its second choice, and $p_0$ as its third choice.

The preferences of vertices in the other 2 gadgets in level~2 corresponding to $c$ ($p^c_t,q^c_t$ for $3 \le t \le 5$ and
$p^c_t,q^c_t$ for $6 \le t \le 8$) are analogous to the above preference lists.
The vertex $s_0$ that appears in $q_0$'s preference list is a vertex from the level~3 gadget corresponding to clause $c$.
Note that  $s_0$ also appears in $q_3$'s preference list and the vertex $t_0$ appears in the preference lists of $p_4$ and $p_7$.

\medskip

\noindent{\em Level~3 vertices.}
Gadgets in level~3 are again clause gadgets. There is exactly one level~3 gadget on 8 vertices $s^c_i,t^c_i$, for $0 \le i \le 3$,
corresponding to clause $c$.  As before, $d_0$ is the last choice of each of these vertices.

\begin{minipage}[c]{0.45\textwidth}
			
			\centering
			\begin{align*}
				&s_0\colon \, t_1  \succ \underline{q_0} \succ t_2\succ \underline{q_3} \succ t_3 \qquad\qquad && t_0\colon \, s_3  \succ \underline{p_7} \succ s_2\succ  \underline{p_4} \succ s_1 \\
			        &s_1\colon \, t_1 \succ t_0  \qquad\qquad && t_1\colon \, s_1 \succ s_0 \\
                                &s_2\colon \, t_2 \succ t_0  \qquad\qquad && t_2\colon \, s_2 \succ s_0 \\
                                &s_3\colon \, t_3 \succ t_0  \qquad\qquad && t_3\colon \, s_3 \succ s_0 \\
			\end{align*}
\end{minipage}

The preference lists of the 8 vertices in the level~3 gadget corresponding to clause $c$ are described above.
For the sake of readability, we have again dropped the superscript $c$ from these vertices. 

It is important to note the preference lists of $s_0$ and $t_0$ here.
Among neighbors in this gadget, $s_0$'s order is $t_1 \succ t_2 \succ t_3$ while
$t_0$'s order is $s_3 \succ s_2 \succ s_1$. Also, $s_0$'s order is interleaved with $q_0 \succ q_3$ (these are vertices from level~2 gadgets) and
$t_0$'s order is interleaved with $p_7 \succ p_4$.

\medskip

There is one more vertex in $G$. This is the vertex $z$, the neighbors of $z$ are $\cup_i \{x_i,y_i\} \cup_c \{a^c_{2i-1},b^c_{2i-1}: i = 1,2,3\} \cup_c \{p^c_{3j+1},q^c_{3j}: j=0,1,2\}$.
The preference order of these neighbors in $z$'s preference list is as follows:
\[z\colon \, x_1  \succ y_1 \succ \cdots \succ x_{n_0} \succ y_{n_0} \succ a^{c_1}_1 \succ b^{c_1}_1 \succ \cdots \]
Here $n_0$ is  the number of variables in $\phi$. Note that
$z$ prefers any neighbor in a level~1 gadget to other neighbors. 

Thus the vertex set $V$ is $\{z\} \cup \{d_0,d_1,d_2,d_3\} \cup_{i=0}^3 \{\mathrm{level}\ i \ \mathrm{vertices}\}$.
We will partition the set $\cup_{i=0}^3\{\mathrm{level}\ i \ \mathrm{vertices}\}$ into $X \cup Y$ where
\begin{eqnarray*}
X & = & \cup_i\{x_i,x'_i\} \cup_c \{a^c_1,\ldots,a^c_6,p^c_0,\ldots,p^c_8,s^c_0,\ldots,s^c_3\}\\
Y & = & \cup_i\{y_i,y'_i\} \cup_c \{b^c_1,\ldots,b^c_6,q^c_0,\ldots,q^c_8,t^c_0,\ldots,t^c_3\}.
\end{eqnarray*}

\begin{lemma}
    \label{new-lemma1}
    For any popular matching $M$ in $G$, the following properties hold:
    \begin{itemize}
    \item[(1)] either $\{(d_0,d_1), (d_2,d_3)\} \subset M$ or $\{(d_0,d_2), (d_1,d_3)\} \subset M$.
    \item[(2)] $M$ matches all vertices in $X \cup Y$.
    \end{itemize}  
\end{lemma}      
\begin{proof}
  Since each of $d_1,d_2,d_3$ is a top choice neighbor for some vertex in $G$, a popular matching in $G$ cannot leave any of these 3 vertices unmatched.
  Since these 3 vertices have no neighbors outside themselves other than $d_0$,
  a popular matching has to match $d_0$ to one of  $d_1,d_2,d_3$.
  Thus $d_0,d_1,d_2,d_3$ are matched among each other in $M$.

  The only possibilities for $M$ when restricted to $d_0,d_1,d_2,d_3$ are the pair of edges $(d_0,d_1), (d_2,d_3)$ or $(d_0,d_2), (d_1,d_3)$.
  The third possibility
  $(d_0,d_3),(d_1,d_2)$ is ``less popular than''  $(d_0,d_1),(d_2,d_3)$ as $d_0,d_2$, and $d_3$ prefer the latter to the former.
  This proves part~(1) of the lemma.

  \smallskip
  
    Consider any vertex $v \in X \cup Y$. If $v$ is left unmatched in $M$ then we either have an alternating path $\rho_1 = (v,d_0)$-$(d_0,d_1)$-$(d_1,d_3)$
    or  an alternating path $\rho_2 = (v,d_0)$-$(d_0,d_2)$-$(d_2,d_1)$ with respect to $M$: in each of these alternating paths, the starting vertex $v$ is
    unmatched in $M$, the middle edge belongs to $M$, and the third edge is a {\em blocking edge} with respect to $M$.

    Suppose $\rho_1$ is an alternating path with respect to $M$. Consider $M \oplus \rho_1$ versus $M$: the vertices $v,d_1,d_3$ prefer $M \oplus \rho_1$
    to $M$ while  $d_0$ and $d_2$ prefer $M$ to $M \oplus \rho_1$; the other vertices are indifferent between $M$ and $M \oplus \rho_1$.
    Thus $M \oplus \rho_1$ is more popular than $M$, a contradiction to $M$'s popularity. Similarly, if $\rho_2$ is an alternating path with respect to $M$
    then $M \oplus \rho_2$ is more popular than $M$. Hence every vertex in $X \cup Y$ has to be matched in $M$. This proves part~(2). \qed
\end{proof}

Since the total number of vertices in $G$ is odd, at least 1 vertex has to be left unmatched in any matching in $G$.
Lemma~\ref{new-lemma1} implies that the vertex $z$ will be left unmatched in $M$.

Let $G_0$ be the subgraph of $G$ induced on $X \cup Y \cup \{z\}$. The matching $M$ restricted to $G_0$ has to be popular on $G_0$, otherwise it would
contradict the popularity of $M$ in $G$. We will now show the following converse of Lemma~\ref{new-lemma1}.

\begin{lemma}
    \label{new-lemma2}
    If $G_0$ admits a popular matching that matches all vertices in $X \cup Y$ then $G$ admits a popular matching.
\end{lemma}
\begin{proof}
  Let $M_0$ be a popular matching in $G_0$ that matches all vertices in $X \cup Y$. We claim $M = M_0 \cup \{(d_0,d_1),(d_2,d_3)\}$ is a popular matching in $G$.

  Let $G'_0$ be the subgraph obtained by
  removing all {\em negative}\footnote{An edge $(u,v)$ is negative to $M_0$ if both $u$ and $v$ prefer their assignments in $M_0$ over each other.} edges
  to $M_0$ from $G_0$. Since $M_0$ is popular in $G_0$, it satisfies the following three necessary and sufficient conditions for
  popularity (from~\cite{HK11}) in $G'_0$.

  \begin{enumerate}
      \item There is no alternating cycle that contains a blocking edge.
      \item There is no alternating path with $z$ as an endpoint that contains a blocking edge.
      \item There is no alternating path that contains {\em two} blocking edges.
  \end{enumerate}
  
  We need to show that $M$ obeys the above 3 conditions in the subgraph $G'$ obtained by deleting negative edges to $M$ from $G$.
  The graph $G'$ is the graph $G'_0$ along with some edges within the gadget $D$.
  There is no edge in $G'$ between $D$ and any vertex in $G_0$ since every edge in $G$ between  $D$ and a vertex in $G_0$ is {\em negative}
  to $M$. This is because for any such edge $(d_0,v)$, the vertex $d_0$ prefers $d_1$ (its partner in $M$) to $v$ and similarly, $v$ prefers each
  of its neighbors in $G_0$ to $d_0$. Since $v \in X \cup Y$, note that $M_0$ matches $v$ to one of its neighbors in $G_0$.

  It is easy to check that the edge set $\{(d_0,d_1),(d_2,d_3)\}$ satisfies the above 3 conditions in the subgraph of $D$ obtained by pruning negative edges
  to $M$. We know that
  $M_0$ satisfies the above 3 conditions in $G'_0$.  Thus $M$ satisfies the above 3 conditions in $G'$. Hence $M$ is popular in $G$. \qed
\end{proof}

We will show the following theorem in Section~\ref{sec:thm3-proof}.

\begin{theorem}
  \label{thm:redn}
  $G_0$ admits a popular matching that matches all vertices in $X \cup Y$ if and only if $\phi$ is 1-in-3 satisfiable.
\end{theorem}  

Since $G$ admits a popular matching if and only if the instance $G_0$ admits a popular matching that matches all vertices in
$X \cup Y$,
Theorem~\ref{thm:redn} implies the NP-hardness of the popular matching problem in a roommates instance $G = (V,E)$. Thus we can
conclude Theorem~\ref{main-thm} stated in Section~\ref{intro}.

\section{Proof of Theorem~\ref{thm:redn}}
\label{sec:thm3-proof}

Our goal now is to use the LP framework for bipartite matchings from Section~\ref{prelims}. However the graph $G_0$ is non-bipartite.
This is due to the presence of the vertex $z$. So let us convert the graph $G_0$ on vertex set $X \cup Y \cup \{z\}$ into a bipartite instance
$H$ by splitting the vertex $z$ into 2 vertices $z$ and $z'$.
That is, every occurrence of $z$ in the preference lists of vertices in $Y$ will be replaced by $z'$.

Thus $H = (X'\cup Y', E_H)$ where $X' = X \cup \{z'\}$ and $Y' = Y \cup \{z\}$. The edge set $E_H$ of $H$ is the same as the edge set of $G_0$,
except that each edge $(z,v)$ where $v \in Y $ gets replaced by the edge $(z',v)$ in $H$.

The graph $H$ is a bipartite graph with $X \cup \{z'\}$ on the left and $Y\cup\{z\}$ on the right.
The preference list of $z$ (similarly, $z'$) is the original preference list of $z$ restricted to neighbors in $X$ (resp., $Y$).
The vertices of $H \setminus \{z,z'\}$ are level~$i$ vertices in $G$, for $i = 0,\ldots,3$. Let $F_H$ be the popular subgraph of $H$.

\begin{lemma}
  \label{lem:conn-comp}
   Let $C$ be any level~$i$ gadget in $H$, where $i \in \{0,1,2,3\}$. All the vertices in $C$ belong to the same connected component in $F_H$.
\end{lemma}
\begin{proof}
  Consider a level~0 gadget in $H$, say on $a^c_1,b^c_1,a^c_2,b^c_2$. The ``men-optimal'' (or $X'$-optimal) stable matching in $H$ contains
  the edges $(a^c_1,b^c_1)$ and $(a^c_2,b^c_2)$ while the ``women-optimal'' (or $Y'$-optimal) stable matching contains the edges $(a^c_1,b^c_2)$ and $(a^c_2,b^c_1)$.
  Thus there are popular edges among these 4 vertices and so these 4 vertices belong to the same connected component in $F_H$.

  Consider a level~1 gadget in $H$, say on $x_i,y_i,x'_i,y'_i$. A stable matching in $H$ contains $(x_i,y_i)$ and $(x'_i,y'_i)$ while
  a dominant matching in $H$ contains $(x_i,y'_i)$ and $(x'_i,y_i)$. Thus there are popular edges among these 4 vertices
  and so these 4 vertices belong to the same connected component in $F_H$.

  Consider a level~2 gadget in $H$, say on $p^c_i,q^c_i$ for $i = 0,1,2$.
  There is a dominant matching in $H$ that contains the edges $(p^c_0,q^c_2)$ and $(p^c_2,q^c_0)$.
  There is also another dominant matching in $H$ that contains the edges $(p^c_1,q^c_2)$ and $(p^c_2,q^c_1)$.
  Thus there are popular edges among these 6 vertices and so these 6 vertices  belong to the same connected component in $F_H$.

  Consider a level~3 gadget in $H$, say on $s^c_i,t^c_i$ for $i = 0,\ldots,3$.
  There is a dominant matching in $H$ that contains $(s^c_0,t^c_1)$, and $(s^c_1,t^c_0)$.
  There is another dominant matching in $H$ that contains $(s^c_0,t^c_2)$ and $(s^c_2,t^c_0)$. 
  There is yet another dominant matching in $H$ that contains $(s^c_0,t^c_3)$ and $(s^c_3,t^c_0)$. 
  Thus there are popular edges among these 8 vertices and so these 8 vertices belong to the same connected component in $F_H$. \qed
\end{proof}  

The lemma below shows that no edge between a level~$\ell$ vertex and a level~$(\ell+1)$ vertex is used in any popular matching in $H$, for $\ell \in \{0,1,2\}$. 

\begin{lemma}
  \label{lem:separate}
  There is no popular edge in $H$ between a level~$\ell$ vertex and a level~$\ell+1$ vertex for $\ell \in \{0,1,2\}$.
\end{lemma}
\begin{proof}
  Let $c = X_i \cup X_j \cup X_k$ be a clause in $\phi$. We will first show that no edge between a level~0 vertex and a level~1 vertex can be popular.
  Consider any such edge  in $H$, say $(a^c_1,y'_j)$. In order to show this edge cannot be present in a popular matching, we will show a popular matching $S$ along
  with a witness $\vec{\alpha}$ such that $\alpha_{a^c_1} + \alpha_{y'_j} > \wt_S(a^c_1,y'_j)$. Then it will immediately follow from the slackness of this edge that
  $(a^c_1,y'_j)$ does not belong to any popular matching (by Lemma~\ref{prop0}).

  Let $S$ be the men-optimal stable matching. The vector $\vec{\alpha} = \vec{0}$ is a witness to $S$. The edges $(a^c_1,b^c_1)$
  and $(x'_j,y'_j)$ belong to $S$, so we have $\wt(a^c_1,y'_j) = -2$ while $\alpha_{a^c_1} = \alpha_{y'_j} = 0$. Thus $(a^c_1,y'_j)$ is not a popular edge.
  We can similarly show that $(x'_k,b^c_1)$ is not a popular edge by considering the women-optimal stable matching $S'$. 
 
  \smallskip
  
  We will now show that no edge between a level~1 vertex and a level~2 vertex is popular.
  Consider any such edge  in $H$, say $(p^c_2,y_j)$.
  Consider the dominant matching $N$ that contains the edges $(p^c_0,q^c_2)$ and $(p^c_2,q^c_0)$.
  All dominant matchings in $H$ contain the edges $(x_j,y'_j)$ and $(x'_j,y_j)$.

  Any witness $\vec{\beta}$ to $N$
  sets $\beta_{p^c_2} = \beta_{q^c_2} = -1$ and $\beta_{x_j} = \beta_{y_j} = 1$. This is because $(x_j,y_j)$ and $(p^c_0,q^c_0)$ are blocking edges to $N$,
  so $\beta_{x_j} = \beta_{y_j} = 1$ and similarly, $\beta_{p^c_0} = \beta_{q^c_0} = 1$ (this makes $\beta_{p^c_2} = \beta_{q^c_2} = -1$).
  Consider the edge $(p^c_2,y_j)$. We have $\wt_N(p^c_2,y_j) = -2$ while  $\beta_{p^c_2} + \beta_{y_j} = 0$. Thus this edge
  is slack and so it cannot be a popular edge. We can similarly show that the edge $(x_k,q^c_2)$ is not popular by considering the dominant matching $N'$
  that includes the edges $(p^c_1,q^c_2)$ and $(p^c_2,q^c_1)$.
 
  \smallskip
  
  We will now show that no edge between a level~2 vertex and a level~3 vertex is popular.
  Consider any such edge  in $H$, say $(s^c_0,q^c_0)$. Consider the dominant matching $T$ that includes the edges $(s^c_0,t^c_1)$ and $(s^c_1,t^c_0)$. Here $q^c_0$ is
  matched either to $p^c_0$ or to $p^c_2$. In both cases, we have $\wt_T(s^c_0,q^c_0) = -2$ while $\gamma_{s^c_0} = -1$ and $\gamma_{q^c_0} = 1$, where
  $\vec{\gamma}$ is a witness to the matching $T$. Hence $(s^c_0,q^c_0)$ is not a popular edge. It can similarly be shown for any edge $e$ 
  between a level~2 vertex and a level~3 vertex in $H$ that $e$ is not a popular edge. \qed
\end{proof}

\subsection{Desired popular matchings in $H$}
It is simple to see that $M$ is a popular matching in $G_0$ that matches all vertices in $X \cup Y$ and leaves $z$ unmatched if and only if
$M$ is a popular matching in $H$ that matches all vertices in $X \cup Y$ and leaves $z$ and $z'$ unmatched.
We will call such a matching $M$ in $H$ a ``desired popular matching'' here. Let $M$ be such a matching and let
$\vec{\alpha} \in\{0,\pm 1\}^n$ be a witness of $M$, where $n = |X'\cup Y'|$.

The following two observations will be important for us. Recall Definition~\ref{def:stab-domn} from Section~\ref{prelims}.

\begin{itemize}
\item[1.] All level~3 gadgets have to be in {\em dominant} state in $\vec{\alpha}$.
\item[2.] All level~0 gadgets have to be in {\em stable} state in $\vec{\alpha}$.
\end{itemize}

  The vertices $s^c_0$ and $t^c_0$, for all clauses $c$, are left unmatched in any stable matching in $H$.
  Since $M$ has to match the unstable vertices $s^c_0$ and $t^c_0$ for all clauses $c$, 
  $\alpha_{s^c_0} = \alpha_{t^c_0} = -1$ for all $c$ (by Lemma~\ref{prop0}). Thus the first observation follows from Lemma~\ref{prop1}. We
  prove the second observation below.

\begin{claim}
  Any level~0 gadget has to be in stable state in $\vec{\alpha}$.
\end{claim}
\begin{proof}
  Consider any level~0 gadget, say on vertices $a^c_3,b^c_3,a^c_4,b^c_4$. Since $M$ is a popular matching, we have $\alpha_{a^c_3} + \alpha_z \ge \wt_M(a^c_3,z)$ and
  $\alpha_{z'} + \alpha_{b^c_3}  \ge \wt_M(z',b^c_3)$. Since $z$ and $z'$ are unmatched in $M$, it follows from Lemma~\ref{prop0} that $\alpha_{z} = \alpha_{z'} = 0$.
  We also have $\wt_M(a^c_3,z) = 0$ since $z$ prefers $a^c_3$ to being unmatched while  $a^c_3$ likes any of its neighbors in $Y$ (one of them is its partner
  in $M$) to $z$. Similarly, $\wt_M(z',b^c_3) = 0$. Thus $\alpha_{a^c_3} \ge 0$ and similarly, $\alpha_{b^c_3} \ge 0$.

  The edge $(a^c_3,b^c_3)$ is a popular edge. Thus $\alpha_{a^c_3} + \alpha_{b^c_3} = \wt_M(a^c_3,b^c_3)$ (by Lemma~\ref{prop0}). Observe that $\wt_M(a^c_3,b^c_3) = 0$
  since either $(a^c_3,b^c_3) \in M$ or $(a^c_3,b^c_4), (a^c_4,b^c_3)$ are in $M$. Thus $\alpha_{a^c_3} + \alpha_{b^c_3} = 0$.
  Since $\alpha_{a^c_3}$ and  $\alpha_{b^c_3}$ are non-negative, it follows that $\alpha_{a^c_3} = \alpha_{b^c_3} = 0$.
  Thus this gadget is in stable state in $\vec{\alpha}$. \qed
\end{proof}

The following lemmas are easy to show and are crucial to our NP-hardness proof. Let $c = X_i \vee X_j \vee X_k$ be any clause in $\phi$.
In our proofs below, we are omitting the superscript $c$ from vertex names for the sake of readability. Recall that $\vec{\alpha} \in \{0, \pm 1\}^n$ is a witness of
our desired popular matching $M$. 
\begin{lemma}
  \label{lemma1}
    For every clause $c$ in $\phi$, at least two of the three level~2 gadgets corresponding to $c$ have to be in dominant state in $\vec{\alpha}$.
\end{lemma}
\begin{proof}
  Let $c$ be any clause in $\phi$.
  We know from observation~1 that the level~3 gadget corresponding to $c$ is in {\em dominant} state in $\vec{\alpha}$. So $\alpha_{s_0} = \alpha_{t_0} = -1$.
  Also, one of the following three cases holds: (1)~$(s_0,t_1)$ and $(s_1,t_0)$ are in $M$, (2)~$(s_0,t_2)$ and $(s_2,t_0)$ are in $M$,
  (3)~$(s_0,t_3)$ and $(s_3,t_0)$ are in $M$.
  
  \begin{itemize}
  \item In case~(1), the vertex $t_0$ prefers $p_4$ and $p_7$ to its partner $s_1$ in $M$. Thus $\wt_M(p_4,t_0) = \wt_M(p_7,t_0) = 0$.
    Since $\alpha_{t_0} = -1$, we need to have $\alpha_{p_4} = \alpha_{p_7} = 1$ so that $\alpha_{p_4} + \alpha_{t_0} \ge \wt_M(p_4,t_0)$ and
    $\alpha_{p_7} + \alpha_{t_0} \ge \wt_M(p_7,t_0)$. Thus the middle and rightmost level~2 gadgets corresponding to $c$
  (see Fig.~\ref{level2:example}) have to be in dominant state in $\vec{\alpha}$.

  \item In case~(2), the vertex $t_0$ prefers $p_7$ to its partner $s_2$ in $M$ and the vertex $s_0$ prefers $q_0$ to its partner $t_2$ in $M$.
    Thus $\alpha_{p_7} = \alpha_{q_0} = 1$ so that $\alpha_{p_7} + \alpha_{t_0} \ge \wt_M(p_7,t_0)$ and $\alpha_{s_0} + \alpha_{q_0} \ge \wt_M(s_0,q_0)$.
    Thus the leftmost and rightmost level~2 gadgets corresponding to $c$ (see Fig.~\ref{level2:example}) have to be in dominant state in $\vec{\alpha}$.

  \item In case~(3), the vertex $s_0$ prefers $q_0$ and $q_3$ to its partner $t_3$ in $M$.
    Thus $\alpha_{q_0} = \alpha_{q_3} = 1$ so that $\alpha_{s_0} + \alpha_{q_0} \ge \wt_M(s_0,q_0)$ and $\alpha_{s_0} + \alpha_{q_3} \ge \wt_M(s_0,q_3)$.
    Thus the leftmost and middle level~2 gadgets corresponding to $c$ (see Fig.~\ref{level2:example}) have to be in dominant state in $\vec{\alpha}$. \qed
  \end{itemize}
\end{proof}

\begin{lemma}
  \label{lemma2}
  For any clause $c$ in $\phi$, {\em at least one} of the level~1 gadgets
  corresponding to variables in $c$ is in dominant state in $\vec{\alpha}$.
\end{lemma}
\begin{proof}
  We showed in Lemma~\ref{lemma1} that
  at least two of the three level~2 gadgets corresponding to $c$ are in dominant state in $\vec{\alpha}$. Assume without loss of generality that these
  are the leftmost gadget and middle gadget (see Fig.~\ref{level2:example}).

  In particular, we know from the proof of Lemma~\ref{lemma1} that
  $\alpha_{q_0} = \alpha_{q_3} = 1$. This also forces  $\alpha_{p_1} = \alpha_{p_4} = 1$. This is because $\alpha_{p_1}$ and $\alpha_{p_4}$ have to be
  non-negative since $p_1$ and $p_4$ are neighbors of the unmatched vertex $z$.

  As $q_0$ and $p_1$ are the most preferred neighbors of $p_2$ and $q_2$, we have $\wt_M(p_2,q_0) = \wt_M(p_1,q_2) = 0$.
  Since $(p_2,q_0)$ and $(p_1,q_2)$ are popular edges, it follows from Lemma~\ref{prop0} that
  $\alpha_{p_2} = \alpha_{q_2} = -1$. Thus either (i)~$(p_2,q_0)$ and $(p_0,q_2)$
  are in $M$ or (ii)~$(p_2,q_1)$ and $(p_1,q_2)$ are in $M$. This means that either $\wt_M(p_2,y_j) = 0$ or $\wt_M(x_k,q_2) = 0$. That is, either $\alpha_{y_j} = 1$
  or $\alpha_{x_k} = 1$.

  Similarly, $\wt_M(p_5,q_3) = \wt_M(p_4,q_5) = 0$ and we can conclude that $\alpha_{p_5} = \alpha_{q_5} = -1$.
  Thus either (i)~$(p_5,q_3)$ and $(p_3,q_5)$ are in $M$ or
  (ii)~$(p_5,q_4)$ and $(p_4,q_5)$ are in $M$. This means that either $\wt_M(p_5,y_k) = 0$ or $\wt_M(x_i,q_5) = 0$. That is, either $\alpha_{y_k} = 1$ or $\alpha_{x_i} = 1$.

  Thus either (i)~the gadgets corresponding to variables $X_i$ and $X_j$ are in dominant state or
  (ii)~the gadget corresponding to $X_k$ is in dominant state in $\vec{\alpha}$.
  Thus at least {\em one} of the level~1 gadgets corresponding to variables in $c$ is in dominant state in $\vec{\alpha}$. \qed
\end{proof}

\begin{lemma}
  \label{lemma3}
   For any clause $c$ in $\phi$, {\em at most one} of the level~1 gadgets corresponding to variables in $c$ is
  in dominant state in $\vec{\alpha}$.
\end{lemma}
\begin{proof}
  We know from observation~2 made at the start of this section that all the three level~0 gadgets corresponding to $c$ are in stable state in $\vec{\alpha}$.
  So $\alpha_{a_t} = \alpha_{b_t} = 0$ for $1 \le t \le 6$.
  Either (i)~$(a_1,b_1)$ and $(a_2,b_2)$ are in $M$ or (ii)~$(a_1,b_2)$ and $(a_2,b_1)$ are in $M$. So either $\wt_M(a_1,y'_j) = 0$ or $\wt_M(x'_k,b_1) = 0$.
  So either $\alpha_{y'_j} \ge 0$ or $\alpha_{x'_k} \ge 0$.

  Consider any variable $X_r$. Either $\{(x_r,y'_r),(x'_r,y_r)\} \subset M$ or $\{(x_r,y_r),(x'_r,y'_r)\} \subset M$. It follows from Lemma~\ref{prop0} that
  $\alpha_{x_r} + \alpha_{y'_r} = \wt_M(x_r,y'_r) = 0$ and  $\alpha_{x'_r} + \alpha_{y_r} = \wt_M(x'_r,y_r) = 0$.
  Also  due to the vertices $z$ and  $z'$, we have $\alpha_{x_r} \ge 0$ and $\alpha_{y_r} \ge 0$.
  Thus $\alpha_{y'_r} \le 0$ and $\alpha_{x'_r} \le 0$. 

  Hence we can conclude that either $\alpha_{y'_j} = 0$ or $\alpha_{x'_k} = 0$. In other words, either the gadget corresponding to $X_j$ or the gadget
  corresponding to $X_k$ is in stable state. Similarly, by analyzing the level~0 gadget on vertices $a_t,b_t$ for $t = 3,4$, we can show that
  either the gadget corresponding to $X_k$ or the gadget corresponding to $X_i$ is in stable state.
  Also, by analyzing the level~0 gadget on vertices $a_t,b_t$ for $t = 5,6$,  either the gadget corresponding to $X_i$ or the gadget corresponding to $X_j$
  is in stable state.

  Thus at least two of the three level~1 gadgets corresponding to variables in clause $c$ are in stable state in $\vec{\alpha}$. Hence at most one of these three gadgets
   is in dominant state in $\vec{\alpha}$. \qed
\end{proof}

\begin{lemma}
  \label{thm1}
If $H$ admits a desired popular matching then $\phi$ has a 1-in-3 satisfying assignment.
\end{lemma}
\begin{proof}
  Let $M$ be a desired popular matching in $H$. That is, $M$ matches all in $X \cup Y$ and leaves $z,z'$ unmatched. Let $\vec{\alpha} \in \{0,\pm 1\}^n$
  be a witness of $M$.
  
  We will now define a $\mathsf{true}$/$\mathsf{false}$ assignment for the variables in $\phi$.
  For each variable $X_r$ in $\phi$ do:
  \begin{itemize}
    \item set $X_r$ to $\mathsf{true}$ if its level~1 gadget is in dominant state in $\vec{\alpha}$, i.e.,
      if $\alpha_{x_r} = \alpha_{y_r} = 1$ or equivalently, $(x_r,y'_r)$ and $(x'_r,y_r)$ are in $M$.
    \item else  set  $X_r$ to $\mathsf{false}$, i.e., here $\alpha_{x_r} = \alpha_{y_r} = 0$ or equivalently, $(x_r,y_r)$ and $(x'_r,y'_r)$ are in $M$.
  \end{itemize}
      
  Since $M$ is our desired popular matching, 
  it follows from Lemmas~\ref{lemma2} and \ref{lemma3} that for every clause $c$ in $\phi$, {\em exactly one} of the three level~1 gadgets
  corresponding to variables in $c$ is in dominant state in $\vec{\alpha}$. That is, for each clause $c$ in $\phi$, exactly one of the three
  variables in $c$ is set to $\mathsf{true}$. \qed 
\end{proof}

\subsection{The converse}

Suppose $\phi$ admits a 1-in-3 satisfying assignment. We will now use this assignment to construct a desired popular matching $M$ in $H$.
For each variable $X_r$ in $\phi$ do:
\begin{itemize}
\item if $X_r = \mathsf{true}$ then include the edges $(x_r,y'_r)$ and $(x'_r,y_r)$ in $M$;
\item else include the edges $(x_r,y_r)$ and $(x'_r,y'_r)$ in $M$.
\end{itemize}
  
Consider a clause $c = X_i \vee X_j \vee X_k$. We know that exactly one of $X_i,X_j,X_k$ is set to $\mathsf{true}$ in our assignment.
Assume without loss of generality that $X_j = \mathsf{true}$. 

We will include the following edges in $M$ from all the gadgets corresponding to $c$.
Corresponding to the level~0 gadgets for $c$ (see Fig.~\ref{level0:example}), we do:
  \begin{itemize}
  \item  Add the edges  $(a^c_1,b^c_1), (a^c_2,b^c_2)$ from the leftmost gadget and $(a^c_5,b^c_6),(a^c_6,b^c_5)$
    from the rightmost gadget to $M$.

    We will select $(a^c_3,b^c_3),(a^c_4,b^c_4)$ from the middle gadget. (Note that we  could also have selected
    $(a^c_3,b^c_4),(a^c_4,b^c_3)$ from the middle gadget.) 
  \end{itemize}

Corresponding to the level~2 gadgets for $c$  (see Fig.~\ref{level2:example}), we do:
  \begin{itemize}
  \item Add the edges $(p^c_0,q^c_0),(p^c_1,q^c_2),(p^c_2,q^c_1)$ from the leftmost gadget,
    $(p^c_3,q^c_3),(p^c_4,q^c_4),(p^c_5,q^c_5)$ from the middle gadget, and
    $(p^c_6,q^c_8),(p^c_7,q^c_7),(p^c_8,q^c_6)$ from the rightmost gadget to $M$.
  \end{itemize}

 Since the leftmost and rightmost level~2 gadgets (see Fig.~\ref{level2:example}) are dominant, we will include $(s^c_0,t^c_2)$ and $(s^c_2,t^c_0)$ in $M$. Hence
  \begin{itemize}
     \item Add the edges $(s^c_0,t^c_2), (s^c_1,t^c_1),(s^c_2,t^c_0),(s^c_3,t^c_3)$ to $M$.
  \end{itemize}

We will show the following theorem now.

\begin{theorem}
  The matching $M$ described above is a popular matching.
\end{theorem}
\begin{proof}
  We will prove $M$'s popularity by describing a witness $\vec{\alpha} \in \{0,\pm 1\}^n$. That is, $\sum_{u\in X'\cup Y'} \alpha_u$ will be 0 and every edge will be
  covered by the sum of $\alpha$-values of its endpoints, i.e., $\alpha_u + \alpha_v \ge \wt_M(u,v)$ for all edges $(u,v)$ in $H$. We will also have
  $\alpha_u \ge \wt_M(u,u)$ for all vertices $u$.

  Set $\alpha_z = \alpha_{z'} = 0$. Also set $\alpha_u = 0$ for all vertices $u$ in gadgets that are in {\em stable} state. That is, there are 
  {\em no blocking edges} to $M$ in these gadgets. This includes all level~0 gadgets,
  and the gadgets in level~1 that correspond to variables set to $\mathsf{false}$, and also the level~2 gadgets in stable state, i.e.,
  such as the gadget with vertices $p^c_3,q^c_3,p^c_4,q^c_4,p^c_5,q^c_5$ (the middle gadget in Fig.~\ref{level2:example}) since we assumed $X_j = \mathsf{true}$.

  \smallskip

  For every variable $X_r$ assigned to $\mathsf{true}$: set $\alpha_{x_r} = \alpha_{y_r} = 1$ and $\alpha_{x'_r} = \alpha_{y'_r} = -1$.
  For every clause, consider the level~2 gadgets corresponding to this clause that are in dominant state:
  for our clause $c$, these are the leftmost and rightmost gadgets in Fig.~\ref{level2:example} (since we assumed $X_j = \mathsf{true}$).

  Recall that we included in $M$ the edges $(p^c_0,q^c_0),(p^c_1,q^c_2),(p^c_2,q^c_1)$ from the leftmost gadget.
  We will set $\alpha_{q^c_0} = \alpha_{p^c_1} = \alpha_{q^c_1} = 1$ and $\alpha_{p^c_0} = \alpha_{p^c_2} = \alpha_{q^c_2} = -1$.
  We also included in $M$ the edges $(p^c_6,q^c_8),(p^c_7,q^c_7),(p^c_8,q^c_6)$ from the rightmost gadget.
  We will set $\alpha_{p^c_6} = \alpha_{q^c_6} = \alpha_{p^c_7} = 1$ and $\alpha_{q^c_7} = \alpha_{p^c_8} = \alpha_{q^c_8} = -1$.

  In the level~3 gadget corresponding to $c$, we included the edges $(s^c_0,t^c_2), (s^c_1,t^c_1),(s^c_2,t^c_0)$, $(s^c_3,t^c_3)$ in $M$.
  We will set $\alpha_{t^c_1} = \alpha_{s^c_2} = \alpha_{t^c_2} = \alpha_{s^c_3} = 1$ and $\alpha_{s^c_0} = \alpha_{t^c_0} = \alpha_{s^c_1} = \alpha_{t^c_3} = -1$.

  The claim below shows that $\vec{\alpha}$ is indeed a valid witness to $M$. Thus $M$ is a popular matching. \qed
 \end{proof}

 \begin{claim}
   The vector $\vec{\alpha}$ defined above is a witness to $M$.
 \end{claim}
 \begin{proof}
  For any edge $(u,v) \in M$, we have $\alpha_u + \alpha_v = 0$, thus $\sum_{u \in X'\cup Y'}\alpha_u = 0$. For any neighbor $v$ of $z$ or $z'$, we have
  $\alpha_v \ge 0$. Thus all edges incident to $z$ or $z'$ are covered by the sum of $\alpha$-values of their endpoints.
  It is also easy to see that for every intra-gadget edge $(u,v)$, we have $\alpha_u + \alpha_v \ge \wt_M(u,v)$. In particular, the endpoints of
  every {\em blocking edge} to $M$ have their $\alpha$-value set to 1. When $X_j = \mathsf{true}$, in the gadgets involving clause $c$,
  $(x_j,y_j),(p^c_1,q^c_1),(p^c_6,q^c_6),(s^c_2,t^c_2)$ are blocking edges to $M$.

  \smallskip
  
  So we will now check that the edge covering constraint holds for all edges $(u,v)$  where $u$ and $v$ belong to different levels. Consider edges in $H$ between a
  level~0 gadget and
  a level~1 gadget. When $X_j = \mathsf{true}$, the edges $(a^c_1,y'_j)$ and $(x'_j,b^c_5)$ are most interesting as they have one endpoint in a gadget in stable state
  and another endpoint in a gadget in dominant state.
  
  Observe that both these edges are {\em negative} to $M$. This is because $a^c_1$ prefers its partner $b^c_1$
  to $y'_j$ and $y'_j$ prefers its partner $x_j$ to $a^c_1$. Thus $\wt_M(a^c_1,y'_j) = -2 < \alpha_{a^c_1} + \alpha_{y'_j}  = 0 - 1$.
  Similarly, $b^c_5$ prefers its partner $a^c_6$ to $x'_j$ and $x'_j$ prefers its partner $y_j$ to $b^c_5$.
   Thus  $\wt_M(x'_j,b^c_5) = -2 < \alpha_{x'_j}  + \alpha_{b^c_5} = - 1 + 0$.

   \smallskip

   We will now consider edges in $H$ between a level~1 gadget and a level~2 gadget. 
   We have $\wt_M(p^c_2,y_j) = 0$ since $p^c_2$ prefers $y_j$ to its partner $q^c_1$ while  $y_j$ prefers its partner $x'_j$ to $p^c_2$. We have
   $\alpha_{p^c_2} + \alpha_{y_j} = -1 + 1 = \wt_M(p^c_2,y_j) = 0$. The edge $(x_k,q^c_2)$ is negative to $M$ and so this is 
   covered by the sum of $\alpha$-values of its endpoints.
   Similarly, $(p^c_8,y_i)$ is negative to $M$ while $\wt_M(x_j,q^c_8) = 0 = 1 - 1 = \alpha_{x_j} + \alpha_{q^c_8}$.
   We have $\wt_M(p^c_5,y_k) = 0$ and $\alpha_{p^c_5} = \alpha_{y_k} = 0$. Similarly, 
   $\wt_M(x_i,q^c_5) = 0$ and $\alpha_{x_i} = \alpha_{q^c_5} =  0$. Thus all these edges are covered.

   \smallskip
   
   We will now consider edges in $H$ between a level~2 gadget and a level~3 gadget. These edges are $(s^c_0,q^c_0), (s^c_0,q^c_3), (p^c_7,t^c_0), (p^c_4,t^c_0)$.
   We have $\wt_M(s^c_0,q^c_0) = 0$ and $\alpha_{s^c_0} = -1, \alpha_{q^c_0} = 1$, so this edge is covered. Similarly, $\wt_M(p^c_7,t^c_0) = 0$ and
   $\alpha_{p^c_7} = 1,\alpha_{t^c_0} = -1$.  The edges $(s^c_0,q^c_3)$ and $(p^c_4,t^c_0)$ are negative to $M$, so they are also covered.  Thus it can be checked
   that $\vec{\alpha}$ is a witness for $M$. \qed
\end{proof}

 Thus $H$ admits a desired popular matching if and only if $\phi$ has a 1-in-3 satisfying assignment.
 This completes the proof of Theorem~\ref{thm:redn}.

 \section{Dominant matchings}
\label{sec:domn}
    Recall that a popular matching $M$ is {\em dominant} if $M$ is more popular than every larger matching. 
  Observe that every popular  matching $M$ in our roommates instance $G = (V,E)$ is a max-size matching:
  this is because $M$ matches all vertices in $G$ except the vertex $z$ (by Lemma~\ref{new-lemma1}).
  Thus every popular  matching in $G$ is dominant and so it follows from Theorem~\ref{main-thm} that the dominant matching problem in $G$ is NP-hard.
  
  Note that the instance $G$ does not admit a stable matching. This is due to the gadget $D = \{d_0,d_1,d_2,d_3\}$.
  However the instance $G_0 = G \setminus D$ admits stable matchings. It is easy to see that a stable matching in $G_0$ matches all vertices in $X \cup Y$
  except the vertices $s^c_0,t^c_0$ for all clauses $c$.
  
  \begin{lemma}
     A popular matching $N$ in $G_0$ is dominant if and only the set of vertices matched in $N$ is $X \cup Y$.
  \end{lemma}
  \begin{proof}
    Let $N$ be any popular matching in $G_0$. Any popular matching has to match all stable vertices in $G_0$ (those matched in any stable matching)~\cite{HK11},
    thus $N$ matches all stable vertices in $G_0$. Suppose some unstable vertex in $X \cup Y$ (say, $s^c_0$) is left unmatched in $N$. 
    We claim that $t^c_0$ also has to be left unmatched in $N$. Since $s^c_1$ and $t^c_1$ have no other neighbors, the edge $(s^c_1,t^c_1) \in N$
    and so there is an augmenting path
    $\rho = s^c_0$-$t^c_1$-$s^c_1$-$t^c_0$ with respect to $N$. Observe that $N$ is {\em not} more popular than $N \oplus \rho$, a larger matching.
    Thus $N$ is not a dominant matching in $G_0$.
    
    In order to justify that $t^c_0$ also has to be left unmatched in $N$,
    let us view $N$ as a popular matching in $H$.
    We know that  $s^c_0$ and $t^c_0$ belong to the same connected component in the popular
    subgraph $F_H$ (by Lemma~\ref{lem:conn-comp}). So if $s^c_0$ is left unmatched in $N$, then $t^c_0$ is also unmatched in $N$ (by Lemma~\ref{prop1}).

    Conversely, suppose $N$ is a popular matching in $G_0$ that matches all  vertices in $X \cup Y$. Then there is no larger matching than $N$ in $G_0$
    and thus $N$ is a dominant matching. \qed
  \end{proof}
  
  Thus a dominant matching exists in $G_0$ if and only if there is a popular matching in $G_0$ that matches all vertices in $X \cup Y$.
  Hence it follows from Theorem~\ref{thm:redn} that the dominant matching problem
  is NP-hard even in roommates instances that admit stable matchings. Thus Theorem~\ref{second-thm} stated in Section~\ref{intro} follows.

\end{document}